\documentclass[11pt]{article}% \documentclass[sigconf,natbib=true]{acmart}

\AtBeginDocument{%
  }

  \usepackage{xcolor}
  \usepackage{graphicx} % Required for inserting images
\usepackage{import}
\usepackage[numbers]{natbib}  % or [round], etc.
\usepackage{amsmath}
\usepackage{subcaption}
\usepackage{tikz}
\usetikzlibrary{positioning,shapes.geometric,arrows.meta}
\usetikzlibrary{shapes, arrows, positioning}
\usepackage[linesnumbered, ruled, vlined]{algorithm2e}
\usepackage{algorithmic}
\usepackage{booktabs}
\usepackage{geometry}
\usepackage{amsmath, amsfonts}
\usepackage{hyperref}
\usepackage{pgf-umlsd}
\usepackage{enumitem}
\usepackage{multirow}
\usepackage{amsthm}

\usepackage{pifont}
\usepackage{bm}
\tikzset{
    every picture/.style={scale=0.7},
    every node/.style={font=\small}
}
\theoremstyle{plain}
\newtheorem{theorem}{Theorem}[section]

\newtheorem{lemma}[theorem]{Lemma}

\theoremstyle{definition}

\newtheorem{approximation}[theorem]{Approximation}

\newcommand{\cmark}{\ding{51}}%

\begin{document}

% \title{PriviRec: Confidential and Decentralized Graph Filtering for Recommender Systems}
\title{Secure Federated Graph-Filtering for Recommender Systems}

\author{
    Julien Nicolas\textsuperscript{1, 2}\textsuperscript{\textdagger},
    César Sabater\textsuperscript{2},
    Mohamed Maouche\textsuperscript{3}, \\
    Sonia Ben Mokhtar\textsuperscript{2},
    Mark Coates\textsuperscript{1} \\
    \\
    \small\textsuperscript{1}ILLS, MILA, McGill University
    \small\textsuperscript{2}CNRS, INSA Lyon, LIRIS \\
    \small\textsuperscript{3}Inria, INSA Lyon, CITI
\\
    \small\textsuperscript{\textdagger}Corresponding author: \texttt{julien.nicolas@insa-lyon.fr}
}

\maketitle

\begin{abstract}
Recommender systems often rely on graph-based filters, such as normalized item-item adjacency matrices and low-pass filters. While effective, the centralized computation of these components raises concerns about privacy, security, and the ethical use of user data. This work proposes two decentralized frameworks for securely computing these critical graph components without centralizing sensitive information. The first approach leverages lightweight Multi-Party Computation and distributed singular vector computations to privately compute key graph filters. The second extends this framework by incorporating low-rank approximations, enabling a trade-off between communication efficiency and predictive performance. Empirical evaluations on benchmark datasets demonstrate that the proposed methods achieve comparable accuracy to centralized state-of-the-art systems while ensuring data confidentiality and maintaining low communication costs. Our results highlight the potential for privacy-preserving decentralized architectures to bridge the gap between utility and user data protection in modern recommender systems.
\end{abstract}

\section{INTRODUCTION}
Recommender systems enable personalized user experiences in services ranging from e-commerce to media streaming and social networks. Over the years, the quest for better performance has led to classical Collaborative Filtering (CF) approaches such as Matrix Factorization (MF) \citep{hu2008collaborative} and Neural Network-based (NN-based) recommender systems, followed by Graph Neural Network (GNN) approaches. More recently, state-of-the-art systems like GF-CF \citep{shen2021powerful} and BSPM \citep{choi2023blurring} achieve superior performance by incorporating graph-based filters such as the normalized item-item matrix and the ideal low-pass filter \citep{he2017neural, he2020lightgcn}. These methods are typically evaluated on large-scale centralized datasets.

However, the wide use of centralized recommender systems raises concerns over data privacy, security, and the monopolization of user data by a few corporations, which could lead to censorship, promotion of some products over others, or risk of private data leaks due to attacks~\citep{agrawal2000privacy, gonzalez2019global}. By decentralizing these systems, it becomes possible to give users more control over their personal data and to prevent it from being concentrated in the hands of a few powerful economic players.

Enforcing privacy in recommender systems is not a new goal: cryptographic techniques can offer strong guarantees but are generally computationally costly~\citep{zhang2020batchcrypt, chai2020secure}, and Differential Privacy (DP) may cause heavy utility loss~\citep{seeman2024between}. Decentralization, which prioritizes keeping data locally rather than distributing it, offers a flexible path to privacy and data ownership and can also be combined with cryptographic or differential privacy techniques for enhanced guarantees. In the context of recommender systems, decentralization allows users to retain control over their personal data while contributing to collaborative computations~\citep{belal2022pepper}. However, implementing decentralized learning presents a number of research challenges, such as ensuring competitive model accuracy, minimizing communication costs, and protecting data privacy.

Federated Learning (FL) \citep{mcmahan2016federated, muhammad2020fedfast, jalalirad2019simple} and Gossip Learning (GL) \citep{hegedHus2021decentralized} emerged to address these challenges. FL allows multiple participants to collaboratively train machine learning models, with the promise of keeping their data local by coordinating the training through a central server that aggregates locally computed updates. This framework has shown promise for simpler models, such as matrix factorization~\citep{chai2020secure}, but scaling it to graph-based recommender systems remains largely unexplored \citep{he2021fedgraphnn, wu2023gnn4fr} as some require knowledge of global graph information during training~\citep{he2020lightgcn, zhang2022iagcn}. GL is a fully decentralized paradigm which does not require a central server. Instead, participants exchange model updates or partial computations with a subset of their peers. Although gossip-based methods are more robust to failures and can scale to large networks \citep{hegedHus2021decentralized}, these advantages come with trade-offs in terms of convergence speed and implementation complexity for models involving global operations like Singular Value Decomposition. As a consequence, most existing decentralized approaches \citep{belal2022pepper, han2004scalable} for recommender systems focus on simple models (e.g., $k$-Nearest Neighbors, Matrix Factorization \citep{hu2008collaborative} or Generalized Matrix Factorization \citep{he2017neural}) and use small data subsets (i.e., $\sim$1000 users \citep{belal2022pepper}), facing scalability challenges in simulation. Hence, while FL and other decentralized solutions are natural candidates for privacy-preserving recommender systems, they struggle to address the complexity of recent powerful methods that rely on global graph-based filtering.

To address this gap, we introduce \textbf{PriviRec} and \textbf{PriviRec-$k$} to decentralize the computation of global graph-based filters in a confidential manner. \textbf{PriviRec} uses Secure Aggregation to collaboratively compute the widely used normalized item-item matrix and the ideal low-pass filter without exposing individual user data. \textbf{PriviRec-$k$} extends this approach by the use of low-rank approximations to reduce communication costs, enabling scalability to larger datasets while maintaining competitive recommendation accuracy.

Our main contributions are as follows:
\begin{enumerate}[leftmargin=*]
\item \textbf{Computations on decentralized data:} We develop algorithms to compute key components of  graph-based recommender systems (See Table~\ref{tab:teaser}) on decentralized data, namely the normalized item-item matrix and ideal low-pass filters.

\item \textbf{Confidentiality:} We use Secure Aggregation for our decentralized computations, and ensure that the results of our computations do not leak confidential user data. To the best of our knowledge, we are the first to decentralize key recommender system components while preserving user confidentiality and incurring no loss in recommendation utility, and keeping a reasonable computational overhead.

\item \textbf{Communication efficiency:}  \textbf{PriviRec-$k$} uses low-rank approximations to offer a trade-off between accuracy and communication costs, addressing a common bottleneck in decentralized systems.

\item \textbf{Empirical benchmarking:} We validate our decentralized methods on standard full-scale recommendation datasets, including Gowalla, Yelp2018, and Amazon-Book, which contain up to 50,000 users. We show that the proposed methods achieve performance comparable to centralized state-of-the-art models.
\end{enumerate}

\section{PRELIMINARIES}
\subsection{Centralized Recommender Systems}

Recommender systems typically rely on a large dataset of user-item interactions held by a single data curator to predict user preferences and recommend interesting items. These systems use algorithms such as matrix factorization, nearest-neighbor approaches, or graph-based methods to model user behavior and item similarities. We present here key concepts to serve as background material for our contributions.

\subsubsection{Graph-Based Representation}
   User-item interaction histories used to train a recommender systems are often modeled as user-item bipartite graphs \cite{he2020lightgcn, shen2021powerful}. For one such graph $\mathcal{G}(\mathcal{V}, \mathcal{E})$, nodes $\mathcal{V}$ represent either users ($\mathcal{U}$) or items ($\mathcal{I}$), and edges $\mathcal{E}$ represent interactions between users and items. The whole graph is typically held by a central server. 

\subsubsection{User-Item interaction matrix $\bm{R}$:}
   This matrix captures user preferences, where each entry $ r_{ui}$ represents the interaction between user $u \in \mathcal{U}$ and item $i \in \mathcal{I}$. $r_{ui}$ may be binary (whether there exists an edge in $\mathcal{E}$ between the nodes associated with $u$ and $i$) or continuous (e.g., ratings). We focus on the binary case and define $r_{ui} = 1$ if user $u$ interacted with item $i$ and $r_{ui} = 0$ otherwise.
\subsubsection{Normalized User-Item interaction matrix $\tilde{\bm{R}}$:}
   The User-Item Interaction Matrix $\bm{R}$ is often normalized to avoid overweighting popular items or popular users. In \citet{choi2023blurring, shen2021powerful}, the normalized users and items degrees matrices are defined as:
   % \begin{equation}
       $\label{U}
       \bm{U} = \text{Diag}(\bm{R} \cdot \bm{1}_{|\mathcal{I}|})$
   % \end{equation}
and
   % \begin{equation}
       $\label{V}
      \bm{V} = \text{Diag}(\bm{1}_{|\mathcal{U}|}^\top  \cdot \bm{R}).$
   % \end{equation}
Then we have:
   \begin{equation}
    \label{normuseritem}
       \tilde{\bm{R}} = \bm{U}^{-\alpha} \bm{R} \bm{V}^{\alpha-1},
   \end{equation}
   where $\alpha$ is a hyperparameter.
\subsubsection{Normalized Item-Item Matrix $ \tilde{\bm{P}} $:}
   The normalized item-item matrix represents similarities or relationships between items based on user interactions. As in \citet{choi2023blurring, shen2021powerful}, it is defined as:  
   \begin{equation}
   \label{normitemitemmat}
      \tilde{\bm{P}} = \tilde{\bm{R}}^\top \tilde{\bm{R}}.
   \end{equation}
   The normalized item-item matrix $ \tilde{\bm{P}} $ can be viewed as an adjacency matrix that models item connectivity.

\subsubsection{Ideal Low Pass Filter $ \bm{F}_{IDL} $}
   This filter is used in graph signal processing to smooth signals (e.g., ratings) over a graph, to keep dominant patterns while suppressing noise. It is defined as:  
   \begin{equation}
   \label{lowpassfilter}
       \bm{F}_{IDL} = \bm{V}^{-\frac{1}{2}} \bm{S}_k \bm{S}_k^\top \bm{V}^{\frac{1}{2}},
   \end{equation}
   where $ \bm{S}_k$ contains the top-$k$ singular vectors of the normalized interaction matrix $ \tilde{\bm{R}}$.
   
\subsubsection{Workflow in Centralized Item-Item Matrix Based Systems}
In the centralized setting, the construction and application of an item-item matrix based recommender system generally proceed as follows:
\begin{enumerate}[leftmargin=*]
\item Graph Construction: Generate a graph representation of the data, where nodes and edges correspond to items and their relationships.
\item Preprocessing: Construct the interaction matrix $\bm{R}$, normalize it according to Eq.~\eqref{normuseritem}, and compute derived matrices like $\tilde{\bm{P}}$ using Eq.~\eqref{normitemitemmat}.
\item Model Training/Fitting: Use algorithms such as matrix factorization or graph neural networks to learn embeddings or predictive models.
\item Recommendation: Predict user-item interactions using the learned model and components such as the item-item, user-item or user-user matrices or the ideal low pass filter. We give a summary of the usage of those components and their performance in Table~\ref{overall_ndcg} in Appendix.
\end{enumerate}

\textbf{In this work}, we focus on recommender models GF-CF \citep{shen2021powerful}, BSPM \citep{choi2023blurring} and Turbo-CF~\citep{park2024turbo}, which offer state-of-the-art recommendation accuracy, scalability and training time. Both GF-CF and BSPM use the normalized item-item matrix $\tilde{\bm{P}}$ and the ideal low pass filter $\bm{F}_{IDL}$ while Turbo-CF only uses $\tilde{\bm{P}}$ (See Table~\ref{tab:teaser}). 

\paragraph{GF-CF} During the inference phase of GF-CF, the predicted interaction matrix $\bm{S}$ is computed as:
\begin{equation}
    \label{GFCFinference}
    \bm{S} = \bm{R} \cdot (\tilde{\bm{P}} + \gamma \cdot \bm{F}_{IDL}),
\end{equation}
where $\gamma$ is a hyperparameter, $\gamma=0.3$ in \citet{shen2021powerful}. $\alpha=\frac{1}{2}$ for the computation of $\tilde{\bm{R}}$.

\paragraph{Turbo-CF} During the inference phase of Turbo-CF, the predicted interaction matrix $\bm{S}$ is computed as:
\begin{equation}
    \label{TurboCFinference}
    \bm{S} = \bm{R} \cdot (\sum_{k=1}^K \alpha_k \overline{\bm{P}}),
\end{equation}
where $\overline{\bm{P}} = \tilde{\bm{P}}^{\circ s}$, with $\circ \cdot$ representing element-wise exponentiation. $s$, $K$, $\alpha$ and $\alpha_k$ are hyperparameters whose optimal values are given by \citet{park2024turbo} depending on the datasets.

\paragraph{BSPM} We refer to the original paper of \citet{choi2023blurring} for the computation of the interaction matrix $\bm{S}$ using BSPM. It is a generalization of Eq.~\eqref{GFCFinference} using Ordinary Differential Equations. The set of hyperparameters can be selected so that its yields exactly Eq.~\eqref{GFCFinference}. As for GF-CF, $\alpha =\frac{1}{2}$ for the computation of $\tilde{\bm{R}}$.

\subsection{Secure Aggregation}
Let $\mathcal{U}$ be a set of users connected to a server and wanting to participate in a computation, with  $|\mathcal{U}| = n$.
Let matrices $ \{\bm{R}^{(u)} | u \in \mathcal{U}\}$ be held locally by users, i.e., $\bm{R}^{(u)}$ is locally held by user $u$. We denote by $\bm{A} = SecAgg(\bm{R}^{(u)}, \mathcal{U})$ the Secure Aggregation of locally held matrices $\bm{R}^{(u)}$ over the set of users $\mathcal{U}$. It is equivalent to computing $\bm{A} = \sum_{u \in \mathcal{U}} \bm{R}^{(u)}$ over a secure channel without revealing each individual contribution $\bm{R}^{(u)}$. In this paper, we use a polylogarithmic-cost version of Secure Aggregation introduced in \citet{bell2020secure} to reduce the communication and computation overhead, although our algorithms are compatible with a wide range of Secure Aggregation protocols (e.g., \citep{so2022lightsecagg, bonawitz2016practical, bonawitz2017practical, bell2020secure}) and decentralized private averaging methods using correlated noise (e.g., \citep{sabater2022accurate, allouah2024privacy}).
% \mohamed{should we put the refs of which method we are talking about?}

\section{PROPOSED METHOD}
We propose collaborative and privacy-preserving methods to compute the normalized item-item matrix (Eq.~\eqref{normitemitemmat}) and the ideal low-pass filter (Eq.~\eqref{lowpassfilter}), which are fundamental components in constructing general filters for recommender systems  (see Table~\ref{overall_ndcg} in Appendix for a review). To enable their use in decentralized settings while preserving user privacy, we propose two variants of our method:

\begin{enumerate}[leftmargin=*]
\item \textbf{PriviRec}: A simple yet efficient approach for computing the necessary components in a privacy-preserving and distributed manner. 
\item \textbf{PriviRec-$k$}: A variant that offers a trade-off between communication cost and recommendation accuracy by utilizing low-rank decompositions.
\end{enumerate}

The components computed using either PriviRec or PriviRec-$k$ can subsequently be used and combined in multiple ways. In our experiments, we use them to instantiate decentralized versions of GF-CF and BSPM to benchmark their utility. In the following subsections, we describe the major steps of our workflows.
\subsection{System model}

Our system consists of a set of $n$ users (clients) and a central coordinating node (server).  
\textbf{Clients} hold local subsets of user-item interactions $\bm{R}^{(u)} \in \mathbb{R}^{|\mathcal{I}|}$, where $\mathcal{I}$ denotes the set of items. Thus, the rows of the entire interaction matrix $\bm{R}$ are distributed across $n$ clients, who never directly share their raw interaction profiles.  
\textbf{The server} coordinates the Secure Aggregation protocols and performs the most computationally expensive tasks such as factorization and normalization. Details of the computation and communication models are presented later. We give an overview of our system in Figure~\ref{fig:system_model}.
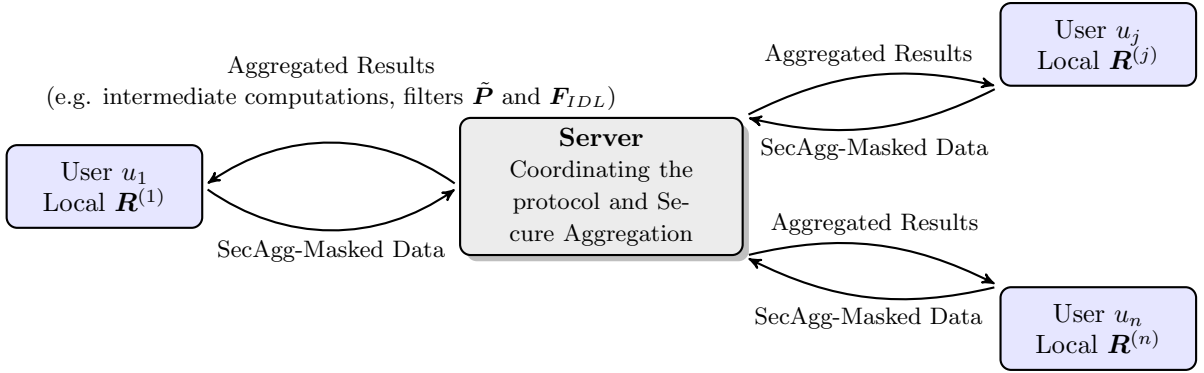
\begin{figure*}
\centering
    % \vspace{-0.3cm}
\begin{tikzpicture}[
  node distance=1.8cm,
  >=stealth',
  thick,
  font=\small,
  rectnode/.style={rectangle, draw, rounded corners, align=center, minimum width=2.6cm, minimum height=1.1cm},
  arrowstyle/.style={->, shorten >=2pt, shorten <=2pt}
]

% -- Server node (center) --
\node[rectnode, fill=gray!15, drop shadow, text width=3.5cm, align=center] (server) 
{ \textbf{Server}\\\footnotesize{Coordinating the protocol and Secure Aggregation} };

% -- Clients around the server --
\node[rectnode, fill=blue!10, left=3.4cm of server, text width=2.0cm] (user1) 
{User $u_1$\\Local $\bm{R}^{(1)}$};
\node[rectnode, fill=blue!10, above right=0.4cm and 3.4cm of server, text width=2.0cm] (user2)
{User $u_j$\\Local $\bm{R}^{(j)}$};
\node[rectnode, fill=blue!10, below right=0.4cm and 3.4cm of server, text width=2.0cm] (userN)
{User $u_n$\\Local $\bm{R}^{(n)}$};

% -----------------------------------------------------------------
% Arrows: from Clients to Server (upload partial info via Secure Agg)
% We curve them with bend left / bend right. Adjust angles to taste.

% user1 -> server
\draw[arrowstyle] (user1.east) to[bend left=-35]
  node[below, align=center, pos=0.5]{\footnotesize SecAgg-Masked Data}
  (server.west);

% user2 -> server
\draw[arrowstyle] (user2.south west) to[bend right=-20] 
  node[below, align=center, pos=0.5]{\footnotesize SecAgg-Masked Data}
  (server.north east);

% userN -> server
\draw[arrowstyle] (userN.north west) to[bend right=-20] 
  node[below, align=center, pos=0.5]{\footnotesize SecAgg-Masked Data}
  (server.south east);

% -----------------------------------------------------------------
% Arrows: from Server back to Clients (broadcast aggregated results)

% Server -> user1
\draw[arrowstyle] (server.west) to[bend right=35] 
  node[above, align=center, pos=0.5]{\footnotesize Aggregated Results \\
  \footnotesize (e.g. intermediate computations, filters $\tilde{\bm{P}} \text{ and } \bm{F}_{IDL}$)\vspace{3mm}} 
  (user1.east);

% Server -> user2
\draw[arrowstyle] (server.north east) to[bend left=20] 
  node[above, align=center, pos=0.5]{\footnotesize Aggregated Results} 
  (user2.south west);

% Server -> userN
\draw[arrowstyle] (server.south east) to[bend left=20] 
  node[above, align=center, pos=0.5]{\footnotesize Aggregated Results} 
  (userN.north west);

\end{tikzpicture}
\vspace{-3mm}
\caption{High-level system model for PriviRec. Each user $u$ holds local interaction data $\bm{R}^{(u)}$. A central server coordinates the overall protocol. It uses Secure Aggregation to receive masked (confidential) partial sums from the clients and homomorphically sums them to compute unmasked global filters  (e.g. $\tilde{\bm{P}}, \bm{F}_{IDL}$). It then \textbf{broadcasts the aggregated results} back to the clients.}
\label{fig:system_model}
\vspace{-3mm}
\end{figure*}
% \mohamed{All figure should be better positioned in the paper not all at the end before references}
\subsection{Threat model}
We consider an \textbf{honest-but-curious threat model}, commonly assumed in federated learning and Secure Multi-Party Computation~\citep{bonawitz2017practical, bell2020secure}, where all parties follow the protocol correctly but may attempt to infer private information. We assume a fixed set of participants, i.e., no dropouts. We want our system to be \textbf{confidential} under this threat model, that is, individual user contributions must not be exposed in the clear during the protocol, and the aggregate must not leak sensitive information. It also needs to be \textbf{correct}: the computed aggregate and outputs must match the non-confidential results as originally defined. Finally, we want the composition of the operations used in the system to be \textbf{secure}, i.e., their composition should not leak unintended information.
\subsection{PriviRec}
PriviRec aims to compute the necessary components for recommendation securely in a distributed setting, where each user keeps their data locally. 

Our approach consists of two main steps, each using polylogarithmic communication and computation cost Secure Aggregation \citep{bell2020secure} to aggregate data from all users without revealing individual contributions. First, we collaboratively compute the item degrees matrix. Then, we compute the item-item matrix and the ideal low-pass filter in a distributed manner. To compute the ideal low-pass filter, we propose a distributed version of the Randomized Power Iteration algorithm, which approximates the leading singular vectors needed for this filter without centralizing any data. 

We summarize the proposed method as a sequence diagram in Figure~\ref{fig:privirec_diagram} in Appendix D and elaborate on the different steps of the algorithm in the following sections.

\subsubsection{Distributed item degrees vector computation:}
\label{itemdegreesvect}
We denote the one-filled column vector of dimension $d$ as $\bm{1}_d$. 
We define $\bm{R}^{(u)}$ as the $u^{\text{th}}$ row of $\bm{R}$, i.e., the part of $\bm{R}$ locally held by user $u$.

 In a distributed and confidential setting, each user $u$ usually only has access to its own interactions $\bm{R}_{u}$. However, it is possible to confidentialy compute the sums of distributed values amongst users
 using Secure Aggregation with low communication and computation overhead (see Subsection \ref{commcostanalysis} for more details). 
By the following, we can compute $\bm{V}$ (defined in Eq.~\eqref{V}) without revealing each individual component $\bm{R}^{(u)}$:
\begin{equation}
    \begin{split}
        \bm{V} &= \text{Diag}(\bm{1}_{|\mathcal{U}|}^\top  \cdot \bm{R}) \\
        &= \text{Diag}(SecAgg(\bm{R}^{(u)}, \mathcal{U})). \\
    \end{split}
\end{equation}

\subsubsection{Distributed item-item gram matrix computation:}
\label{itemitemmatrix}
\begin{theorem}
    We can compute $\tilde{\bm{P} }$ (defined in Eq.~\eqref{normitemitemmat}) in a distributed and secure setting as:
\begin{equation}
    \begin{split}
        \tilde{\bm{P} } = \bm{V}^{-\frac{1}{2}}(SecAgg(\frac{1}{d_{user}(u)} {\bm{R}^{(u)}} ^\top\bm{R}^{(u)}, \mathcal{U})) \bm{V}^{-\frac{1}{2}}.\\
    \end{split}
\end{equation}
\end{theorem}
\begin{proof}
    The item-item normalized adjacency matrix is defined as:
\begin{equation}
\begin{split}
\tilde{\bm{P} } & =\tilde{\bm{R}}^\top\tilde{\bm{R}} \\
% & = ( \bm{U}^{-\frac{1}{2}} \bm{R}\bm{V}^{-\frac{1}{2}})^\top( \bm{U}^{-\frac{1}{2}} \bm{R} \bm{V}^{-\frac{1}{2}}) \\
&=  \bm{V}^{-\frac{1}{2}} (\bm{U}^{-\frac{1}{2}}\bm{R} )^\top(\bm{U}^{-\frac{1}{2}} \bm{R} )\bm{V}^{-\frac{1}{2}}. \\
\end{split}
\end{equation}

We then define $\bm{P}' = (\bm{U}^{-\frac{1}{2}}\bm{R} )^\top(\bm{U}^{-\frac{1}{2}} \bm{R} )$, $\bm{P}'_{ij}$ as the element of $\bm{P}'$ at line $i$ and column $j$ and $d_{user}(u) =  \sum_{i=0}^{|\mathcal{I}|-1} r_{ui}$. This leads to the following expression for the $(i,j)$-th element of $\bm{P}'$: 
\begin{equation}
\begin{split}
\bm{P}'_{ij} &= ((\bm{U}^{-\frac{1}{2}}\bm{R} )^\top(\bm{U}^{-\frac{1}{2}} \bm{R} ))_{ij} \\
% & = ((\bm{U}^{-\frac{1}{2}}\bm{R})^T)_{i,*}( \bm{U}^{-\frac{1}{2}}\bm{R})_{*,j} \\
&=  ((\bm{U}^{-\frac{1}{2}}\bm{R})_{*,i})^\top( \bm{U}^{-\frac{1}{2}}\bm{R})_{*,j} \\
% &=  \sum_{u=0}^{N-1} \frac{1}{\sqrt{d_{user}(u) }}r_{ui}\frac{1}{\sqrt{d_{user}(u) }}r_{uj} \\
&=  \sum_{u=0}^{N-1} \frac{1}{d_{user}(u)} \cdot r_{ui} \cdot r_{uj}. \\
\end{split}
\end{equation}

By noticing that $(\bm{R} ^\top\bm{R} )_{ij} =  \sum_{u} r_{ui} \cdot r_{uj}$, we can deduce that $\bm{P}' = \sum_{u} \frac{1}{d_{user}(u)} {\bm{R}^{(u)}} ^\top \bm{R}^{(u)}$.
We can therefore compute $\tilde{\bm{P} }$ in a distributed and privacy-preserving setting:
\begin{equation}
    \begin{split}
        \tilde{\bm{P} } &= \bm{V}^{-\frac{1}{2}}\bm{P}' \bm{V}^{-\frac{1}{2}}\\
        &= \bm{V}^{-\frac{1}{2}}(\sum_{u} \frac{1}{d_{user}(u)} {\bm{R}^{(u)}} ^\top \bm{R}^{(u)}) \bm{V}^{-\frac{1}{2}}\\
        &= \bm{V}^{-\frac{1}{2}}(SecAgg(\frac{1}{d_{user}(u)} {\bm{R}^{(u)}} ^\top \bm{R}^{(u)}, \mathcal{U})) \bm{V}^{-\frac{1}{2}}. \qedhere
    \end{split}
\end{equation}
\vspace{-0.4cm}
\end{proof}
\subsubsection{Distributed ideal low pass filter computation:}
\label{ilpf}

The ideal low pass filter is computed using the singular vectors $\bm{S}_k$ of the normalized interaction matrix (Eq.~\eqref{computeilpf}). In the centralized recommender system literature \citep{choi2023blurring, shen2021powerful}, these vectors are approximated using  a centralized version of the randomized power method (Algorithm 4.4 in \citet{halko2011finding}).

With Theorem~\ref{computeilpf} below, we propose a decentralized and secure adaptation of the randomized power method. This version is similar to Algorithm 1 of \citet{2024matrixfactorisation}, although it is less susceptible to numerical instabilities (see Remark 4.3 in \citet{halko2011finding}) and uses the Gram-Schmidt QR procedure to output an approximation of the leading eigenvectors and not of their range.  Differentially private variants of this procedure for positive semi-definite matrices have also been explored by \citet{hardt2014noisy}, 
\citet{balcan2016improved}, and 
\citet{nicolas2024differentially}. 

\begin{algorithm}[ht]
% \small
	\begin{algorithmic}
	\caption{Distributed secure power method}
    \label{alg:decentralizedpowermethod}
%	\SetAlgoLined
	\STATE{\textbf{Input}: distributed matrices  $\{\bm A^{(u)} | u \in \mathcal{U}\}$, number of iterations $L$, target rank $k$, iteration rank $p\geq k$.}
	\STATE{\textbf{Output}: approximated eigen-space $\bm X_L\in\mathbb R^{m\times p}$, with orthonormal columns, upper triangular matrix $\bm{T}_\ell$.}
	\STATE{\textbf{Initialization}}: orthonormal $\bm X_0\in\mathbb R^{m\times p}$ by QR decomposition on a random Gaussian matrix $\bm G_0$;\\
    \STATE The central node broadcasts $\bm X_{0}$ to all $n$ computing nodes;\\
    \STATE Computing node $u$ computes $\bm Y^{(u)}_0={\bm A^{(u)}}^\top \bm X_{0}$.\\
    \STATE The central node computes with the clients $\bm Y_0= SecAgg({\bm Y_0^{(u)}},  \mathcal{U})$. \\
    \STATE The central node computes QR factorization $\bm Y_0=\bm X_0\bm T_0$.
    \FOR{$\ell=1$ to $L-1$}
	\STATE{
        1. The central node broadcasts $\bm X_{\ell-1}$ to all $n$ computing nodes;\\
		2. Computing node $u$ computes $\bm Y^{(u)}_\ell={\bm A^{(u)}}^\top \bm A^{(u)}\bm X_{\ell-1}$.\\
		3. The central node computes with the clients $\bm Y_\ell= SecAgg({\bm Y_\ell^{(u)}},  \mathcal{U})$. \\
  		4. The central node computes QR factorization $\bm Y_\ell=\bm X_\ell\bm T_\ell$. \\
	}
	\ENDFOR
	\end{algorithmic}
\end{algorithm}

\begin{theorem}
\label{computeilpf}

Using Algorithm~\ref{alg:decentralizedpowermethod}, we can collaboratively compute the ideal low-pass filter $\bm{F}_{IDL}$. We provide the proof in Appendix A. Our approach is also depicted as a sequence diagram in Figure~\ref{fig:secure_power_method}. 
\end{theorem}

\begin{figure}[!ht]
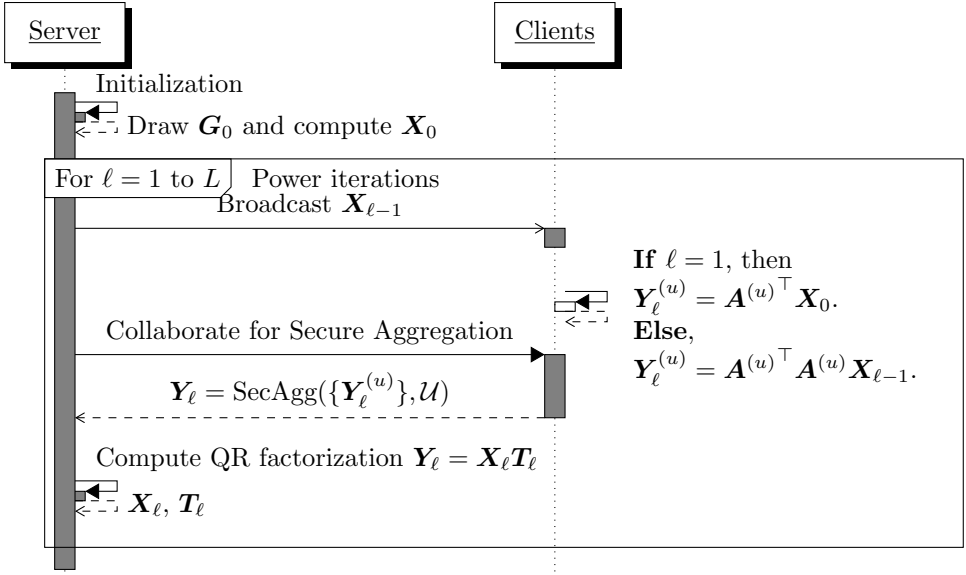

% \vspace{-0.3cm}
  % \hspace{-10mm}
  % \scalebox{0.9}
  {
  \begin{sequencediagram}
    \newthread[gray]{A}{Server}{}
    \newinst[7]{B}{Clients}{}

    \begin{call}{A}{Initialization}{A}{Draw $\bm{G}_0$ and compute $\bm{X}_0$}
    \end{call}
    % Iterative process
    \begin{sdblock}{For $\ell = 1$ to $L$}{Power iterations}
     % \postlevel
    \begin{messcall}{A}{Broadcast $\bm{X}_{\ell-1}$}{B}{}
        \end{messcall}
            % \postlevel
        % \begin{call}{B}{Compute}{B}{ $\bm{Y}_\ell^{(u)} = 
        % \newline
        % \postlevel
        % \begin{cases}
        % {\bm{A}^{(u)}}^\top \bm{X}_0 &\text{if $l=1$}\\
        % {\bm{A}^{(u)}}^\top \bm{A}^{(u)} \bm{X}_{\ell-1} &\text{else}
        % \end{cases}
        % $}
        % \end{call}
        \begin{call}{B}{}{B}{
        \begin{tabular}{l}
            \textbf{If} $\ell = 1$, then \\
            $\bm{Y}_\ell^{(u)} = {\bm{A}^{(u)}}^\top \bm{X}_0$. \\ 
            \textbf{Else},\\
            $\bm{Y}_\ell^{(u)} = {\bm{A}^{(u)}}^\top \bm{A}^{(u)} \bm{X}_{\ell-1}$.\\
        \end{tabular}
        % \hspace{-5mm}
        % \textbf{if} $\ell = 1$, then\\
        % \quad \\
        % \textbf{else},\\
        % \quad $\bm{Y}_\ell^{(u)} = {\bm{A}^{(u)}}^\top \bm{A}^{(u)} \bm{X}_{\ell-1}.$
        }
        \end{call}

        % \begin{call}{B}{Compute}{B}{
        %   Compute $\bm{Y}_\ell^{(u)}$ via: 
        %   \[
        %     \bm{Y}_\ell^{(u)} = 
        %     \begin{cases}
        %       {\bm{A}^{(u)}}^\top \bm{X}_0, & \ell = 1,\\
        %       {\bm{A}^{(u)}}^\top \bm{A}^{(u)} \bm{X}_{\ell-1}, & \ell > 1,
        %     \end{cases}
        %   \]
        % }
        % \end{call}
% \postlevel
            % \postlevel
        \begin{call}{A}{Collaborate for Secure Aggregation}{B}{$\bm{Y}_\ell = \text{SecAgg}(\{\bm{Y}_\ell^{(u)}\}, \mathcal{U})$}
        \postlevel
        \end{call}
            \postlevel
        \begin{call}{A}{Compute QR factorization $\bm{Y}_\ell = \bm{X}_\ell \bm{T}_\ell$}{A}{$\bm{X}_\ell$, $\bm{T}_\ell$}
        \end{call}
    \end{sdblock}
  \end{sequencediagram}
  }
      % \vspace{-0.3cm}
\caption{Sequence diagram for the Distributed Secure Power Method algorithm. The server and clients collaboratively compute the eigen-space of $\tilde{\bm{P}}$.}
    % \vspace{-0.2cm}
\label{fig:secure_power_method}
\end{figure}

\subsection{PriviRec-$k$}
We now propose a variant that employs low-rank decomposition to reduce the communication cost of PriviRec. This version is also distributed and privacy-preserving but allows one to find a trade-off between communication cost and recommendation utility. PriviRec-$k$ has 3 main steps: 
\begin{enumerate}[leftmargin=*] 
\item The distributed computation of the item degrees matrix (using polylogarithmic-cost Secure Aggregation \cite{bell2020secure}), as in PriviRec. 
\item The distributed computation of the top-$k$ singular vectors and values of $\tilde{\bm{P}}$ (using Secure Aggregation and a distributed version of Randomized Power Iteration \cite{2024matrixfactorisation, halko2011finding}).
\item The computation of the ideal low-pass filter and the approximated item-item matrix using the resulting top-$k$ singular vectors and values. 
\end{enumerate}
We summarize our approach as a sequence diagram in Figure~\ref{fig:privireck_diagram} in Appendix D and detail the steps in the following sections.

\subsubsection{Communication efficient item-item matrix computation:}
\label{comeffitemitem}
As introduced in Section \ref{itemitemmatrix}, the normalized item-item matrix is defined as:
\begin{equation}
\begin{split}
\vspace{-3mm}
\tilde{\bm{P} }  =\tilde{\bm{R}}^\top\tilde{\bm{R}} 
                 =\sum_u \tilde{\bm{R}}_u^\top\tilde{\bm{R}}_u. 
\end{split}
\vspace{-3mm}
\end{equation} 
Aggregating $\tilde{\bm{R}}_u^\top\tilde{\bm{R}}_u \in \mathbb{R}^{|\mathcal{I}|\times|\mathcal{I}|}$, $\forall u \in \mathcal{U}$ can be communication-intensive. To reduce this cost, we develop this approximation:
\begin{approximation}
    We can approximate $\tilde{\bm{P}} = \bm{S}_k Diag(\bm{T}_{L-1}) \bm{S}_k^\top$, where $\bm{S}_k$ and $\bm{T}_{L-1}$ are the outputs of the Distributed Secure Power method executed on the distributed $\bm{R}^{(u)}$.
\end{approximation}
Thus, we only need to communicate $\bm{S}_k \in  \mathbb{R}^{|\mathcal{I}| \times k}$ and the diagonal $Diag(\bm T_{L-1})$ using this decomposition. The communication savings arises because $Diag(\bm T_{L-1}) \in \mathbb{R}^{k}$ as opposed to $\tilde{\bm{P}} \in \mathbb{R}^{|\mathcal{I}|\times|\mathcal{I}|}$. The matrix $\bm{S}_k$ already needs to be broadcasted to compute the Ideal Low Pass Filter. We show in Section~\ref{results} that it is possible to retain a competitive NDCG while having $k \ll |\mathcal{I}|$.
\begin{proof}[Derivation]
    Using the decentralized power method, clients collaborate to compute and approximation of the top-$k$ right singular vectors $\bm{S}_k$ of $\tilde{\bm{R}}$. This yields the following low-rank approximation:
\[
\tilde{\bm{R}}_k^\top \approx \bm{S}_k \bm{\Sigma}_k \bm{Z}_k^\top,
\]
where $\bm{\Sigma}_k$ is rectangular diagonal, $\bm{S}_k$ contains the top-$k$ right singular vectors and $\bm{Z}_k$ contains the top-$k$ left singular vectors from $\tilde{\bm{R}}_k$. We use a low rank approximation $\tilde{\bm{P}} \approx \tilde{\bm{P}_k}$ where:
\begin{align}
     \tilde{\bm{P}}_k &\approx \tilde{\bm{R}}_k^\top \tilde{\bm{R}}_k \\
     &= \bm{S}_k \Lambda_k \bm{S}_k^\top,
\end{align}
where $\Lambda_k \in \mathbb{R}^{k \times k}$ is a diagonal matrix.

At steps 2 and 3 of iteration $L-1$ of Algorithm~\ref{alg:decentralizedpowermethod}, we compute $\bm{Y}_{L-1} = {\bm A}^\top \bm A\bm X_{L-2}$. By injecting $\bm A = \tilde{\bm{R}}^\top$ and using the approximation $X_{L-2} \approx \bm{S}_k$, we arrive at the following:
\begin{align}
    \bm{Y}_{L-1} &= {\bm A}^\top \bm A\bm X_{L-2} \\
    & \approx  \tilde{\bm{P}} \bm{S}_k\\
     &\approx \bm{S}_k \Lambda_k.
\end{align}

Then at step 4, the central node computes the QR factorization $\bm{Y}_{L-1} = \bm X_{L-1}\bm T_{L-1}$, 
and we have approximately $X_{L-1} \approx \bm{S}_k$.
Therefore by left multiplication $\Lambda_k \approx \bm T_{L-1}$ and since $\Lambda_k$ should be diagonal:
\begin{align}
\vspace{-3mm}
    \Lambda_k = Diag(\Lambda_k) \approx  Diag(\bm T_{L-1}). \quad \qedhere
\end{align}
\vspace{-8mm}
\end{proof}
\subsection{Security analysis}

\begin{lemma}
\label{secanallemma}
    PriviRec and PriviRec-$k$ are confidential under the honest-but-curious threat model, i.e., respect the following principles:
    \begin{enumerate}[leftmargin=*]
    \item \textbf{Confidentiality}: Individual user contributions ($\bm{R}^{(u)}$) must not be exposed in the clear during the protocol, and the aggregate must not leak sensitive information.
    \item \textbf{Correctness}: The computed aggregate (e.g., $\bm{V}', \bm{P}'$) and outputs (e.g., $\tilde{\bm{P}}, \bm{F}_{IDL}$) must match the expected results as defined by the algorithm.
    \item \textbf{Secure Composition}: If all individual operations in the protocol are secure, and their composition does not leak unintended information, then the overall protocol is secure.
\end{enumerate}
\end{lemma}

\begin{proof}[Proof of Lemma~\ref{secanallemma}.]
    The Secure Aggregation protocol is a well-established technique in Multi-Party Computation which guarantees that:
    \begin{enumerate}[leftmargin=*]
        \item Individual contributions from users ($\bm{R}^{(u)}$) are never revealed; only the aggregate result (e.g., $\sum_{u \in \mathcal{U}} \bm{R}^{(u)}$) is accessible. This is achieved through masking techniques and secret sharing~\cite{bell2020secure}.
        \item It is resistant to collusion up to a predefined threshold of adversarial participants \cite{bell2020secure}.
        \item The final output is correct and matches the expected aggregate sum, assuming that honest participants provide valid inputs.
    \end{enumerate} 

We now describe why the computations of the components of PriviRec and PriviRec-$k$ respect the security principles:
\subsubsection{Item Degrees Computation} The computation of the item degrees matrix $\bm{V}' = \text{SecAgg}(\{\bm{R}^{(u)}\})$ relies on Secure Aggregation, which ensures that individual contributions are masked, and only the aggregate is revealed. This step does not expose any user-specific data and is not sensitive because the aggregated item degrees vector only informs about the popularity of items.

\subsubsection{Item-Item Matrix Computation} The computation of
\begin{equation}
    \bm{P}' = \text{SecAgg}(\{\frac{1}{d_{\text{user}}(u)} \bm{R}^{(u)\top} \bm{R}^{(u)}\})
\end{equation}
also leverages Secure Aggregation. As in the previous step, individual user contributions ($\bm{R}^{(u)\top} \bm{R}^{(u)}$) remain confidential, and only the aggregated result is used, which represents the similarity between items but does not give information about users. 

\subsubsection{Ideal Low-Pass Filter Computation} The decentralized power method involves iterative computations using intermediate results aggregated through Secure Aggregation:
% \begin{equation}
    $\bm{Y}_\ell = \text{SecAgg}(\{\bm{Y}_\ell^{(u)}\}).$
% \end{equation}
These intermediate results are aggregated across users and do not reveal individual contributions. Furthermore, the final outputs ($\bm{S}_k$, $\bm{F}_{IDL}$) only depend on aggregated data and randomized computations. $\bm{S}_k$ is a projector for rank $k$ approximations of the item-item matrix computation, which is not sensitive information.

\subsubsection{Security of Composition}
Secure Aggregation is secure under composition with the honest-but-curious (also known as semi-honest) threat model \citep{bell2020secure}. This argument has been used to allow the use of Secure Aggregation for Federated Learning \citep{bonawitz2017practical, liu2022privacy, secagg45808}. Other threats arise in this context due to the fact that different overlapping set of users can be used at different iterations, allowing differencing attacks. As stated, we assume that the same set of users participate in each iteration of the protocol, i.e., there are no dropouts and participants are honest.

We showed that no step reveals user-specific data in the clear. Moreover, the compositional security of Secure Aggregation ensures that combining these steps does not introduce vulnerabilities. Thus, the proposed protocol achieves both confidentiality and correctness under the stated threat model, protecting user privacy throughout its execution.
\end{proof}

\subsection{Communication cost analysis}
% \mohamed{I'm not sure why table 4 is in the appendix it seems that it should be in the paper}
\label{commcostanalysis}
We denote by $n$ the number of users participating in the collaborative filtering process. Table~\ref{tbl:secureagg_custom}
presents the communication cost introduced by Secure Aggregation \citep{bell2020secure}. The table compares PriviRec, PriviRec-k, and federated learning based algorithms.

\begin{table}[!h]
\small
    % \vspace{-0.4cm}
    \centering
    \caption{Communication overhead complexity ($O(\cdot)$) of using Secure Aggregation (F. GCNs/MFs denotes federated versions of most GCN or MF based recommender systems)}
        \vspace{-0.2cm}
    \label{tbl:secureagg_custom}
    \begin{tabular}{lcc}\toprule
             & \textbf{Client} & \textbf{Server} \\ \midrule
        PriviRec     & $L \log (n) + |\mathcal{I}|^2$   & $L n \log(n) + n|\mathcal{I}|^2$  \\
        PriviRec-$k$     & $L \log (n) + L k_2 \cdot |\mathcal{I}|$   & $L n \log(n) + L   k_2  n|\mathcal{I}|$  \\
        F. GCNs/MFs     & $e \log (n) +  e k_3 (|\mathcal{I}|+n)$   & $e n \log(n) + e k_3  n(|\mathcal{I}|+n)$  \\
        \bottomrule
    \end{tabular}
            \vspace{-2mm}
\end{table}
\begin{proof}[Communication overhead calculations]
We analyze the communication overhead introduced by Secure Aggregation in our methods by breaking down the different communications.

For the \textbf{PriviRec} method, Secure Aggregation is performed on the item degrees vector (size $|\mathcal{I}|$), the item-item matrix (size $|\mathcal{I}| \times |\mathcal{I}|$), and on intermediate results from the randomized power iteration algorithm  (each of size $|\mathcal{I}| \times k_1$, where $k_1$ is the number of singular vectors), over $L$ iterations (typically $L=2-5$). 

For the \textbf{PriviRec-$k$} variant, Secure Aggregation is used to compute the item degrees vector (of size $|\mathcal{I}|$), on intermediate results from the randomized power iteration over $L$ iterations and on the approximated singular values and vectors of the item-item matrix. The matrices exchanged are of sizes $|\mathcal{I}| \times k_2$, where $k_2$ is the number of singular vectors used to compute both the low-rank approximation of the item-item matrix and the ideal low pass filter. For the Gowalla and Yelp2018 datasets, $k_2$ can be set in the order of 2000 to have NDCG values competitive with PriviRec, while for those datasets $|\mathcal{I}|$ is fixed to $40,981$ and $38,048$ respectively.

% which have a size significantly smaller than $|\mathcal{I}| \times |\mathcal{I}|$ due to the low-rank approximation.

In contrast, federated versions of graph convolutional networks (GCNs) or gradient-descent-based matrix factorization methods require more communication rounds because they train over $e$ epochs (typically $e=1000$). They need to communicate user and item embeddings of dimension $k_3$ in each epoch, resulting in a total communication size of $(|\mathcal{I}| + n) \times k_3 \times e$.

For one iteration of Secure Aggregation, the communication complexity for each client is \( O(\log(n) + l) \) (Section 3.4 of \citet{bell2020secure}), where \( n \) is the number of clients and \( l \) is the length of the message being communicated. In this protocol, each client communicates with \( k = O(\log(n)) \) other clients to exchange the necessary information for mask generation and aggregation. Specifically, these operations include multiple exchanges to coordinate the generation of masks, the transmission of masked inputs of size \( O(l) \), and the verification of consistency in the exchanged data.

On the server side, the communication complexity is \( O(n(\log n + l)) \), as the server interacts with \( n \) clients to aggregate the inputs and coordinate the secure computation process. By incorporating the number of iterations of Secure Aggregation required for the protocol and the sizes of the matrices involved, we derive the communication complexities presented in Table~\ref{tbl:secureagg_custom}.
% For one iteration of Secure Aggregation, the communication complexity for each client is \( O(\log (n) + l) \) (Section 3.4 of \citet{bell2020secure}), where $n$ is the number of clients,  and $l$ is the length of the message. In this protocol, $k= O(\log(n))$ is the number of clients in contact with each client. Specifically, clients engage in up to \( 5k \) key agreements, transmit \( 2k \) masked shares, exchange masked input data of size \( O(l) \), and verify or provide up to \( 5k \) signatures. On the server side, the communication complexity is \( O(n(\log n + l)) \), as the server communicates with \( n \) clients. By injecting the number of iterations of Secure Aggregation and the matrices sizes, we derive the communication complexities presented in Table~\ref{tbl:secureagg_custom}.
\end{proof}
As shown in Table~\ref{tbl:secureagg_custom}, the communication overhead for both clients and the server is significantly lower when using \textbf{PriviRec-$k$} compared to \textbf{PriviRec}, because $ L \cdot k_2 \ll |\mathcal{I}| $. Moreover, computing the filters using PriviRec-$k$ requires less total communication than regular decentralized gradient descent methods used in GCNs or matrix factorization and in particular much less communication bursts (which introduce latency). Indeed, the only iterative process in PriviRec-$k$ is the decentralized randomized power iteration, which is more communication-efficient than gradient descent for computing matrix factorizations. Typically, PriviRec-$k$ requires a small number of iterations (e.g., $ L $ in the order of 2 to 5), whereas gradient descent methods may require a much larger number of iterations (e.g., $ e = 1000 $ \citep{he2020lightgcn}). 

For the Gowalla dataset, we use use $L=3$ for \textbf{PriviRec} and \textbf{PriviRec-$k$} and use $e=1000$ for the Federated GCN methods as advised by \citet{he2020lightgcn}. 
% For this dataset, we have $k_1=256$, $k_2=2000$ and $k_3=64$. Using the communication overheads complexities from Table~\ref{tbl:secureagg_custom} with these values and the statistics of the Gowalla dataset yields client costs of 1,679M for PriviRec, 246M for PriviRec-$k$ and 4,534M for the Federated GCNs. This analysis demonstrates that PriviRec-$k$ not only preserves confidentiality but also offers significant communication efficiency advantages over traditional decentralized methods. 
For this dataset, we have $k_1=256$, $k_2=2000$ and $k_3=64$. Using the communication overheads complexities from Table~\ref{tbl:secureagg_custom} with these values and the statistics of the Gowalla dataset yields dominant terms for client costs of 1,679M floats for PriviRec, 246M floatsfor PriviRec-$k$ and 4,534M floats for the Federated GCNs. Although we provide complexity expressions, the mapping to the actual number of floats transmitted for each algorithm is approximately the same. This indicates that PriviRec-$k$ not only preserves confidentiality but also offers significant communication efficiency advantages over traditional decentralized methods. 

% \mohamed{did you put "privacy" here rather than confidentiallity on purpose?}
\section{EXPERIMENTS}
In this section, we evaluate the impact of computing the filters of State of the Art methods GF-CF and BSPM with decentralized methods and the impact of the number of factors $k$ on the recommendation performance of GF-CF. 
\subsection{Datasets}
\begin{table}[htb]
\vspace{-0.2cm}
    \centering
    \caption{Statistics of datasets}\label{tbl:data}
    \vspace{-0.3cm}
    \begin{tabular}{lcccc}\toprule
        \textbf{Dataset}     & \textbf{Users} & \textbf{Items} & \textbf{Interactions} & \textbf{Density} \\ \midrule
        Gowalla     & 29,858  & 40,981  & 1,027,370      & 0.084\% \\
        Yelp2018    & 31,668  & 38,048  & 1,561,406      & 0.130\% \\
        Amazon-Book & 52,643  & 91,599  & 2,984,108      & 0.062\% \\
        \bottomrule
    \end{tabular}
    \vspace{-0.3cm}
\end{table}
We benchmark our approaches on the widely adopted Gowalla \cite{cho2011friendship}, Yelp2018 \cite{he2020lightgcn} and Amazon-Book \cite{wang2019neural} datasets, using the standard train/test splits to compare fairly with the existing literature. We show the dataset statistics in Table~\ref{tbl:data}.

\subsection{Metrics}

We use two common ranking metrics to benchmark the performances of the model: Recall@20 and NDCG@20. These widely-employed metrics measure the quality of top-$k$ item ranking  \citep{jarvelin2002cumulated, he2020lightgcn, choi2023blurring, shen2021powerful}.  

{\textbf{Recall}} measures the proportion of relevant items successfully retrieved among the top 20 recommendations, which makes it particularly suitable for tasks where completeness is critical.  

{\textbf{NDCG@20}} (Normalized Discounted Cumulative Gain) is used to evaluate the quality of ranking by assigning higher scores to relevant items appearing earlier in the recommended list. This metric provides an assessment of both relevance and ranking order.  

\subsection{Experimental setup}

We evaluate the impact of computing the filters of State of the Art methods GF-CF and BSPM with decentralized methods. Specifically, we compute the NDCG@20 on Gowalla, Yelp2018 and Amazon-Book to get the performance of BSPM and GF-CF when using decentralized components computed with \textbf{PriviRec}. Table~\ref{tbl:main_exp} reports results for the setting $L=2$.

We also study the impact of the number of factors $k$ on recommendation performance of GF-CF when using \textbf{PriviRec-$k$}. We vary $k$ between 256 and 3584 in increments of 256 and set $L=2$.
% \mohamed{is yelp really on the appendix? is it a typo?}

\subsection{Results}
\label{results}
\subsubsection{Impact of decentralization}
We can see in Table~\ref{tbl:main_exp} that the performances of GF-CF and BSPM with components computed with PriviRec are similar to those we obtain with their centralized counterparts.
This allows us to confirm experimentally that PriviRec can be used in recommender systems offering competitive results while allowing users to keep their data locally, strengthening user privacy. We note that the decentralized versions are theoretically equivalent to the centralized versions. Minor differences in NDCG can be attributed to the randomness of the Randomized Power Method and the numerical instabilities arising from performing a single QR decomposition per iteration after multiplication by \( A^\top A \). Although algebraically equivalent, the centralized implementation uses separate QR decompositions after the multiplication by \( A \) and \( A^\top \).

\begin{table}[htb]
\centering
  \vspace{-0.25cm}
\caption{Comparison of NDCG@20 Scores for Centralized (Centr.) and PriviRec-Based Models (Priv.) on Gowalla, Yelp2018 and Amazon-Book.}
  \vspace{-0.3cm}
\label{tbl:main_exp}
\begin{tabular}{l | c c  |  c c  | c c}
\toprule 
\textbf{Model} & \multicolumn{2}{c}{\textbf{Gowalla}} & \multicolumn{2}{c}{\textbf{Yelp}} & \multicolumn{2}{c}{\textbf{Amazon}} \\ 
& Centr. & Priv. & Centr. & Priv. & Centr. & Priv. \\ 
\midrule 
GF-CF & 0.1518 & 0.1528 & 0.0571 & 0.0561 & 0.0584 & 0.0584 \\
Turbo-CF & 0.1531 & 0.1531 & 0.0574 & 0.0574 & 0.0611 & 0.0611 \\
BSPM-LM & 0.1570 & 0.1580 & 0.0584 & 0.0576 & 0.0610 & 0.0610 \\
BSPM-EM & 0.1597 & 0.1595 & 0.0593 & 0.0589 & 0.0609 & 0.0609 \\
\bottomrule
\end{tabular}
\vspace{-0.2cm}
\end{table}

\subsubsection{Impact of the number of factors $k$}

We illustrate the importance of the number of factors $k$ by plotting recommendation performance of GF-CF when using \textbf{PriviRec-$k$} on Gowalla in Figure~\ref{fig:impactkperf} and on Yelp2018 in Figure~\ref{fig:impactkperfyelp} . 
For Gowalla, we can see that PriviRec-$k$ offers competitive NDCGs compared to the regular PriviRec for $k$ in the range $[2048, 3584]$. In fact, it even has higher NDCG than the baseline. It also offers competitive NDCG (although slightly inferior) on Yelp2018 when $k$ ranges between 2816 and 3584. To determine an appropriate value of $k$ for a particular system, one approach is to optimize over a validation set by running PriviRec-$k$ (and eventually PriviRec in parallel) for different values of $k$. This allows the selection of the smallest $k$ that achieves the desired performance while minimizing communication costs. 

% \begin{figure*}
% \centering
% \begin{subfigure}{.55\textwidth}
%   \centering
%     \includegraphics[width=1.0\linewidth]{gowalla_privireck_ndcg_comparison.pdf}
%     \caption{Gowalla}
%     \vspace{-0.2cm}
%     \label{fig:impactkperf}
% \end{subfigure}%
% \begin{subfigure}{.55\textwidth}
%   \centering
%     \includegraphics[width=1.0\linewidth]{yelp2018_privireck_ndcg_comparison.pdf}
%     \caption{Yelp2018}
%     \vspace{-0.2cm}
%     \label{fig:impactkperfyelp}
% \end{subfigure}
% \caption{Comparison of the NDCG of GF-CF on Gowalla (left) and Yelp2018 (right), when its components are either computed using PriviRec, or with PriviRec-$k$. We vary $k$ between 256 and 3584 in increments of 256.}
% \label{fig:test}
% \end{figure*}
% \begin{figure}[ht] % or [!ht]
% \centering
%   \vspace{-0.3cm}
% \begin{subfigure}{\linewidth}
%   \centering
%   \includegraphics[width=0.95\linewidth]{gowalla_privireck_ndcg_comparison.pdf}
%   \vspace{-1.8mm}
%   \caption{Gowalla}
%   \vspace{-0.3cm}
%   \label{fig:impactkperf}
% \end{subfigure}
% \\[1em] % gap between subfigures
% \begin{subfigure}{\linewidth}
%   \centering
%   \includegraphics[width=0.95\linewidth]{yelp2018_privireck_ndcg_comparison.pdf}
%   \vspace{-1.8mm}
%   \caption{Yelp2018}
%   \vspace{-0.2cm}
%   \label{fig:impactkperfyelp}
% \end{subfigure}
% \caption{NDCG of GF-CF on Gowalla and Yelp2018, when its components are either computed using PriviRec, or PriviRec-$k$. We vary $k$ between 256 and 3584 in increments of 256.}
%         \vspace{-4mm}
% \label{fig:test}
% \end{figure}

\begin{figure}[ht] % or [!ht]
\centering
  % \vspace{-0.3cm}
\begin{subfigure}{0.49\linewidth}
  \centering
  \includegraphics[width=\linewidth]{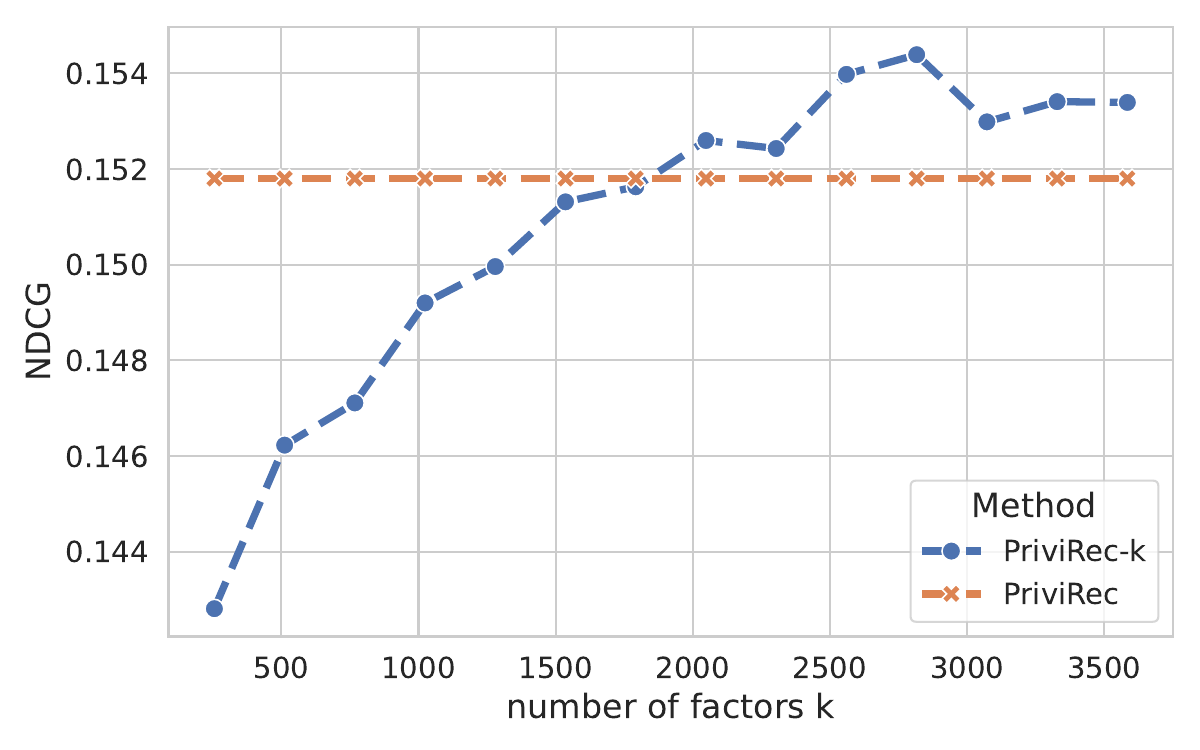}
  \vspace{-3mm}
  \caption{Gowalla}
  % \vspace{-0.3cm}
  \label{fig:impactkperf}
\end{subfigure}
\hfill % Adds horizontal space between the subfigures
\begin{subfigure}{0.49\linewidth}
  \centering
  \includegraphics[width=\linewidth]{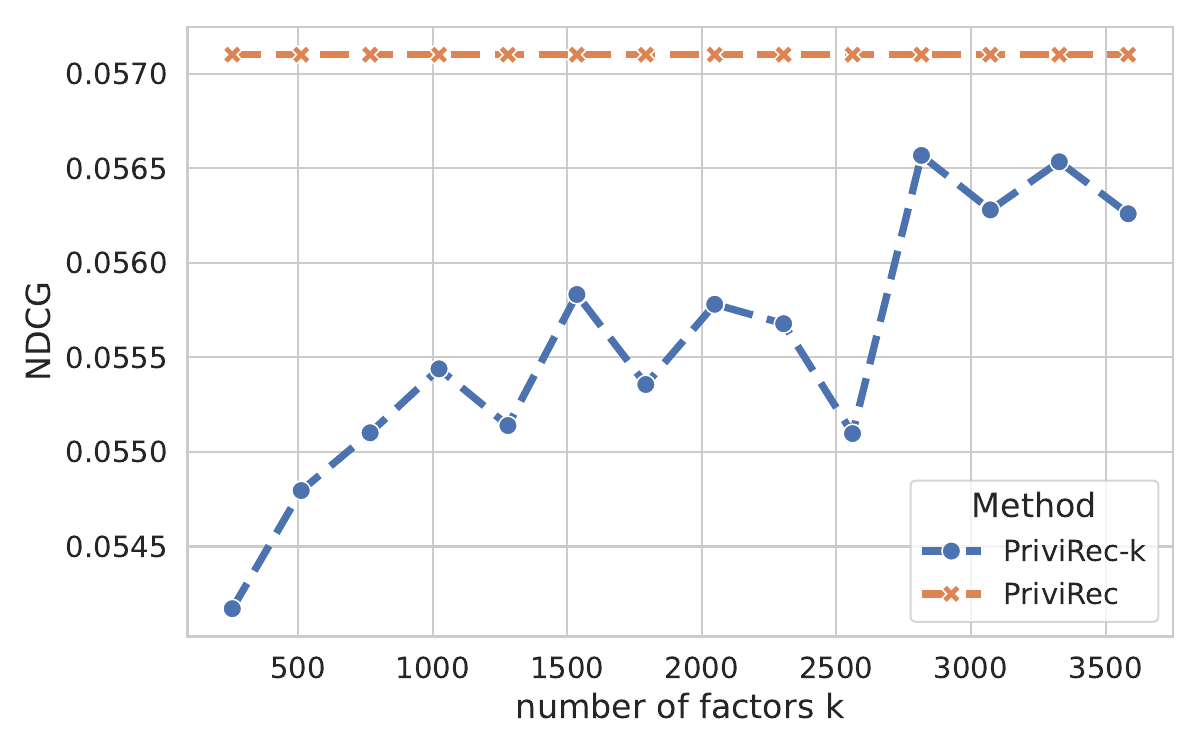}
  \vspace{-3mm}
  \caption{Yelp2018}
  % \vspace{-0.2cm}
  \label{fig:impactkperfyelp}
\end{subfigure}
\caption{NDCG of GF-CF on Gowalla and Yelp2018, when its components are either computed using PriviRec, or PriviRec-$k$. We vary $k$ between 256 and 3584 in increments of 256.}
\vspace{-4mm}
\label{fig:test}
\end{figure}

\section{CONCLUSION \& FUTURE WORK}
We proposed PriviRec and PriviRec-$k$, two decentralized and secure frameworks for computing essential components of recommender systems, including the normalized item-item matrix and the ideal low-pass filter. By using Secure Aggregation and the distributed randomized power method, our methods strengthen user privacy while allowing competitive recommendation accuracy. PriviRec-$k$ leverages low-rank approximations to reduce communication overhead, providing an efficient trade-off between communication cost and recommendation utility.

In future work, we aim to adapt our methods to larger datasets and higher-dimensional embeddings. We could indeed explore factoring the embeddings using tensor trains, or use sparse versions of Secure Aggregation which can lower the communication overheads without sacrificing on privacy or accuracy. We also want to explore more formal notions of user privacy using Differential Privacy.

\clearpage
\clearpage
\bibliographystyle{abbrvnat}
\bibliography{sample-base}

%%
%% If your work has an appendix, this is the place to put it.
\appendix
\clearpage
\section{Security of the decentralized randomized power iteration}
We now provide the proof of Theorem~\ref{computeilpf}.
\begin{proof}
    Algorithm~\ref{alg:decentralizedpowermethod} allows us the equivalent of the following using Secure Aggregation:
\begin{equation}
\begin{split}
\bm{Y}  &= (\tilde{\bm{R}}^\top \tilde{\bm{R}})^{\alpha} \tilde{\bm{R}}^\top \bm{\Omega} \\
        &= \bm{Q} \bm{T}, \\
\end{split}
\end{equation} such that $\bm{Y}$ is in the range of $\tilde{\bm{R}}^\top $ and $\bm{Q} \bm{T}$ is a  Gram-Schmidt QR decomposition of $\bm{Y}$, i.e. $\bm{Q}$ is an orthonormal matrix (i.e., $\bm{Q}^{\top} \bm{Q} = \bm{I}$) and $\bm{T}$ is an upper triangular matrix. 
Here $\alpha$ is a positive integer chosen as hyper-parameter and $\bm{\Omega}$ a random Gaussian Matrix. \\

According to \cite{halko2011finding}, we can then compute a rank K approximation of $\tilde{\bm{R}}^\top$, which will approximately share its top K singular vectors:
\begin{equation}
\begin{split}
\tilde{\bm{R}}^\top_K  &= \bm{Q}\bm{Q}^\top \tilde{\bm{R}}^\top. \\
\end{split}
\end{equation} 
We also define $ \bm{B} =  \bm{Q}^\top \tilde{\bm{R}}^\top $. If we define the SVD of $\bm{B}$ as $\bm{B}  = \bm{Z} \bm{\Sigma} \bm{M}^\top$, then we can deduce the SVD of $\tilde{\bm{R}}^\top_K$ indirectly as:
\begin{equation}
\begin{split}
\tilde{\bm{R}}^\top_K &= \bm{Q}\bm{Z} \bm{\Sigma} \bm{M}^\top. \\
\end{split}
\end{equation}
We therefore have:
\begin{equation}
\begin{split}
\tilde{\bm{R}}_K  &= \bm{M}\bm{\Sigma}^\top\bm{Z}^\top  \bm{Q}^\top \\
                &= \bm{M}\bm{\Sigma}^\top \bm{S}^\top. \\ 
\end{split}
\end{equation}

We recognize $\bm{S} = \bm{Q}\bm{Z}$ as the right singular vectors of $\tilde{\bm{R}}_K$, therefore as the top-K right singular values of $\tilde{\bm{R}}$.
We are only interested in the $\bm{S} \bm{S}^\top $ product, we do not need $\bm{S}$ itself. We notice that:
\begin{equation}
\begin{split}
\bm{S} \bm{S}^\top  &=  \bm{Q}\bm{Z}\bm{Z}^\top  \bm{Q}^\top \\
                 &=  \bm{Q}\bm{\mathcal{I}}\bm{Q}^\top \text{($\bm{Z}$ orthonormal by def.)}\\
                 &=  \bm{Q}\bm{Q}^\top.\\
\end{split}
\end{equation} 
It is therefore sufficient to apply the distributed version of the Randomized Power Iteration algorithm on $\tilde{\bm{R}}$ and to compute a QR decomposition to be able to compute $\bm{F}_{IDL}$ ($\bm{V}$  was computed in Sec. \ref{itemitemmatrix}) and we do not need to compute the full SVD of $\tilde{\bm{R}}$.
\end{proof}

\section{Use of the item-item matrix and the ideal low pass filter}
We use results of Table 4 of \citet{choi2023blurring} to show an NDCG@20 comparison across models on Gowalla, Yelp2018 and Amazon-Book, and additionally state whether those models use the item-item matrix and/or the Ideal Low Pass Filter defined previously. We present these results in Table~\ref{overall_ndcg}. We can see that most competitive models use the item-item gram matrix, and that some use the ideal low pass filter. 

\begin{table}[htb]
    \centering
    \caption{\label{overall_ndcg} NDCG@20 Comparison Across Models with numerical results from \citep{choi2023blurring}. We indicate models leveraging $\tilde{\bm{P}}$ and $\bm{SS^\top}$ with checkmarks. Yelp is short for Yelp2018 and Amazon for Amazon-Book.}
    % \resizebox{0.45\textwidth}{!}
    {
    \begin{tabular}{lccccc}
        \toprule
        \textbf{Model} & \textbf{Gowalla} & \textbf{Yelp} & \textbf{Amazon} & \textbf{$\tilde{\bm{P}}$} & \textbf{$\bm{SS^\top}$} \\
        \midrule
        EASE\textsuperscript{R} \citep{steck2019EASE} & 0.1467 & 0.0552 & 0.0567 &  &  \\
        YouTubeNet \citep{Covington2016YoutubeNet}           & 0.1473 & 0.0567 & 0.0388 &  &  \\
        MF-CCL \citep{mao2021simplex}              & 0.1493 & 0.0572 & 0.0447 &  &  \\
        \midrule
        GAT \citep{velickovic2018GAT}            & 0.1236 & 0.0431 & 0.0235 &  &  \\
        JKNet \citep{xu2018jknet}            & 0.1391 & 0.0502 & 0.0343 &  &  \\
        DropEdge \citep{rong2020dropedge} & 0.1394 & 0.0506 & 0.0270 & \cmark &  \\
        APPNP \citep{klicpera2019appnp}    & 0.1462 & 0.0521 & 0.0299 & \cmark &  \\
        DisenGCN \citep{ma2019DisenGCN}      & 0.1174 & 0.0454 & 0.0254 &  &  \\
        \midrule
        LR-GCCF \citep{chen20LRGCCF}        & 0.1452 & 0.0498 & 0.0296 & \cmark &  \\
        LightGCN \citep{he2020lightgcn}    & 0.1554 & 0.0530 & 0.0315 & \cmark &  \\
        SGL-ED \citep{Wu2021SGLED}          & 0.1539 & 0.0555 & 0.0379 & \cmark &  \\
        DeosGCF \citep{liu2020deoscillated}        & 0.1477 & 0.0504 & 0.0316 &  &  \\
        IMP-GCN \citep{liu2021IMP-GCN}      & 0.1567 & 0.0531 & 0.0357 & \cmark &  \\
        BUIR\textsubscript{NB} \citep{lee2021BUIR} & 0.1301 & 0.0526 & 0.0346 &  &  \\
        DGCF \citep{Xiang2020DGCF19}             & 0.1561 & 0.0534 & 0.0324 &  &  \\
        IA-GCN \citep{zhang2022iagcn}            & 0.1562 & 0.0537 & 0.0373 &  &  \\
        UltraGCN \citep{Mao21UltraGCN}      & 0.1580 & 0.0561 & 0.0556 & \cmark &  \\
        SimpleX \citep{mao2021simplex}          & 0.1557 & 0.0575 & 0.0468 &  &  \\
        LT-OCF \citep{choi2021ltocf}         & 0.1574 & 0.0549 & 0.0341 & \cmark &  \\
        HMLET \citep{kong2022hmlet}          & 0.1589 & 0.0557 & 0.0371 & \cmark &  \\
        LinkProp \citep{fu2022revisiting}          & 0.1477 & 0.0559 & 0.0559 &  &  \\
        LinkProp-Multi \citep{fu2022revisiting}    & 0.1573 & 0.0571 & 0.0588 &  &  \\
        MGDCF \citep{hu2022mgdcf}            & 0.1589 & 0.0572 & 0.0378 & \cmark &  \\
        GTN \citep{fan2022GTN}                 & 0.1588 & 0.0554 & 0.0346 &  &  \\
        \midrule
        GF-CF \citep{shen2021powerful}     & 0.1518 & 0.0571 & 0.0584 & \cmark & \cmark \\
        Turbo-CF \citep{park2024turbo} & 0.1531 & 0.0574 & \textbf{0.0611} & \cmark &  \\
        BSPM-LM \citep{choi2023blurring}    & 0.1570 & 0.0584 & {0.0610} & \cmark & \cmark \\
        BSPM-EM \citep{choi2023blurring}    & \textbf{0.1597} & \textbf{0.0593} & {0.0609} & \cmark & \cmark \\
        \bottomrule
    \end{tabular}
    }
    \label{tab:teaser}
\end{table}

% \section{Communication overhead due to Secure Aggregation}

% Table 4 provides an analysis of the communication complexity of the PriviRec and PriviRec-k procedures, together with a comparison to federated versions of 

\section{Sequence diagrams}

Figure 4 provides a sequence diagram for the PriviRec algorithm, showing the computations and communications used to return the approximate item-item matrix and ideal low-pass filter. Figure 5 shows the equivalent sequence diagram for the PriviRec-k algorithm. 

\begin{figure*}
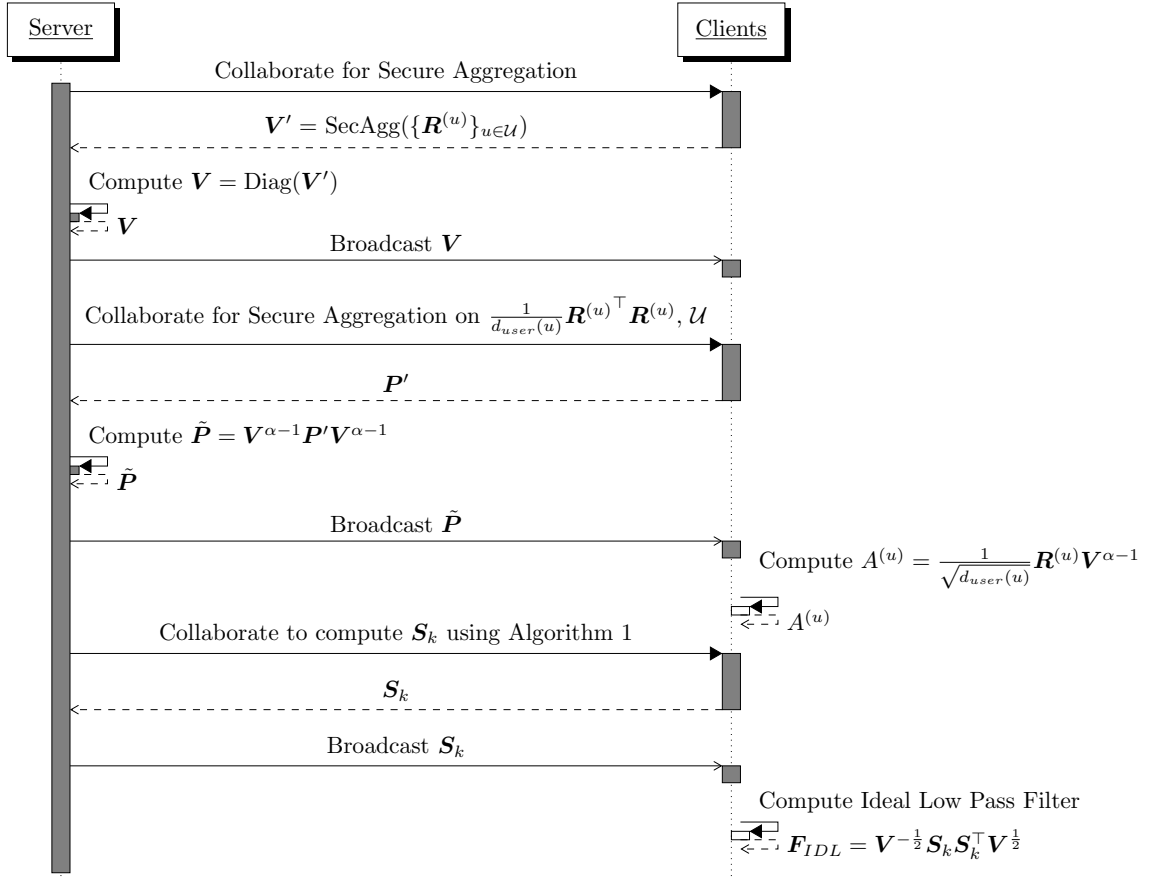

  \centering
  \scalebox{0.89}{
  \begin{sequencediagram}
    \newthread[gray]{A}{Server}{}
    \newinst[12]{B}{Clients}{}

    % Step 1: SecAgg computation for V'
    \begin{call}{A}{Collaborate for Secure Aggregation}{B}{$\bm{V}'=$ SecAgg($\{ \bm{R}^{(u)} \}_{u \in \mathcal{U}}$)}
           \postlevel
    \end{call}
        \postlevel
      % \postlevel
    % Step 2: Compute and broadcast V
    \begin{call}{A}{Compute $\bm{V} = \text{Diag}(\bm{V}')$}{A}{$\bm{V}$}
    \end{call}
          % \postlevel
    \begin{messcall}{A}{Broadcast $\bm{V}$}{B}{}
    \end{messcall}
        \postlevel

    % Step 3: SecAgg computation for P'
    \begin{call}{A}{Collaborate for Secure Aggregation on $\frac{1}{d_{user}(u)} {\bm{R}^{(u)}}^\top \bm{R}^{(u)}$, $\mathcal{U}$}{B}{$\bm{P}'$}
           \postlevel
    \end{call}
        \postlevel
      % \postlevel
    % Step 4: Compute and broadcast approximate item-item matrix
    \begin{call}{A}{Compute $\tilde{\bm{P}} = \bm{V}^{\alpha-1} \bm{P}' \bm{V}^{\alpha-1}$}{A}{$\tilde{\bm{P}}$}
    \end{call}
          \postlevel
    \begin{messcall}{A}{Broadcast $\tilde{\bm{P}}$}{B}{}
    \end{messcall}
        % \postlevel

    % Step 5: Clients compute A^(u) locally
    \begin{call}{B}{Compute $A^{(u)} = \frac{1}{\sqrt{d_{user}(u)}} \bm{R}^{(u)} \bm{V}^{\alpha-1}$}{B}{$A^{(u)}$}
    \end{call}
        % \postlevel

    % Step 6: Decentralized Power Method
    \begin{call}{A}{Collaborate to compute $\bm{S}_k$ using Algorithm~\ref{alg:decentralizedpowermethod}}{B}{$\bm{S}_k$}
           \postlevel
    \end{call}
        \postlevel

    % Step 7: Broadcast leading singular vectors
    \begin{messcall}{A}{Broadcast $\bm{S}_k$}{B}{}
    \end{messcall}
        % \postlevel

    % Step 8: Clients compute Ideal Low Pass Filter
    \begin{call}{B}{Compute Ideal Low Pass Filter}{B}{$\bm{F}_{IDL} = \bm{V}^{-\frac{1}{2}} \bm{S}_k \bm{S}_k^\top \bm{V}^{\frac{1}{2}}$}
    \end{call}

  \end{sequencediagram}
  }
\caption{Sequence diagram for the PriviRec Algorithm. The algorithm involves secure aggregation, server and clients-side computations to return the approximate item-item matrix $\tilde{\bm{P}}$ and the ideal low-pass filter $\bm{F}_{IDL}$.}
\label{fig:privirec_diagram}
\end{figure*}

% \subsection{PriviRec-$k$}
% We provide a sequence diagram for our proposed PriviRec-$k$ method as Figure~\ref{fig:privireck_diagram}.
\begin{figure*}
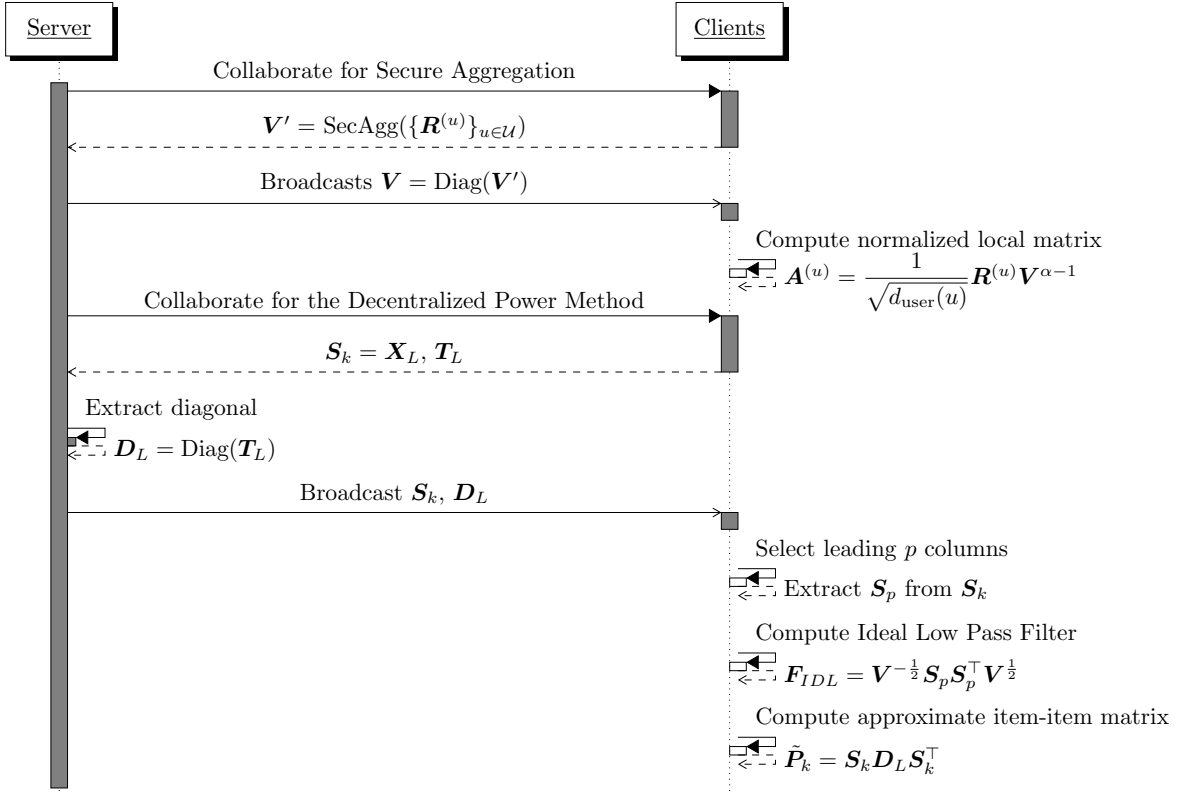

\small
  \centering
  \scalebox{0.89}{
  \begin{sequencediagram}
    \newthread[gray]{A}{Server}{}
    \newinst[12]{B}{Clients}{}

    % Step 1: Collaborative Computation (SecAgg)
    \begin{call}{A}{Collaborate for Secure Aggregation}{B}{$\bm{V}' =$ SecAgg($\{ \bm{R}^{(u)} \}_{u \in \mathcal{U}}$)}
           \postlevel
    \end{call}
        \postlevel
    % Step 2: Server computes normalized item degree matrix
    \begin{messcall}{A}{Broadcasts $\bm{V} = \text{Diag}(\bm{V}')$}{B}
    \end{messcall}
        % \postlevel
    % Step 3: Clients compute normalized local matrix
    \begin{call}{B}{Compute normalized local matrix}{B}{$\bm{A}^{(u)} = \dfrac{1}{\sqrt{d_{\text{user}}(u)}} \bm{R}^{(u)} \bm{V}^{\alpha -1}$}
    \end{call}
        % \postlevel
    % Step 4: Collaborative Computation (Decentralized Power Method)
    \begin{call}{A}{Collaborate for the Decentralized Power Method}{B}{$\bm{S}_k = \bm{X}_L$, $\bm{T}_L$}
           \postlevel
    \end{call}

        \postlevel
    \begin{call}{A}{Extract diagonal}{A}{$\bm{D}_L = \text{Diag}(\bm{T}_L)$}
    \end{call}
            \postlevel
    \begin{messcall}{A}{Broadcast $\bm{S}_k$, $\bm{D}_L$}{B}{}
    \end{messcall}
        % \postlevel
    % Step 6: Clients compute final matrices
    \begin{call}{B}{Select leading $p$ columns}{B}{Extract $\bm{S}_p$ from $\bm{S}_k$}
    \end{call}
            \postlevel
    \begin{call}{B}{Compute Ideal Low Pass Filter}{B}{ $\bm{F}_{IDL} = \bm{V}^{-\frac{1}{2}} \bm{S}_p \bm{S}_p^\top \bm{V}^{\frac{1}{2}}$}
    \end{call}
            \postlevel
    \begin{call}{B}{Compute approximate item-item matrix}{B}{$\tilde{\bm{P}}_k = \bm{S}_k \bm{D}_L \bm{S}_k^\top$}
    \end{call}

  \end{sequencediagram}
  }
\caption{Sequence diagram for the PriviRec-$k$ Algorithm. The algorithm involves collaborative computation, server and client-side computation to return the approximate item-item matrix $\tilde{\bm{P}}_k$ and the ideal low-pass filter $\bm{F}_{IDL}$.}
\label{fig:privireck_diagram}
\end{figure*}
% \begin{figure*}[htb]
%     \centering
%     \includegraphics[width=0.6\linewidth]{yelp2018_privireck_ndcg_comparison.pdf}
%     \caption{Comparison of the NDCG of GF-CF on Yelp2018, when its components are either computed using PriviRe, or with PriviRec-$k$. We vary $k$ on a range from 256 to 3584, with increments of 256.}
%     \label{fig:impactkperfyelp}
% \end{figure*}

\end{document}